\documentclass[journal]{IEEEtran}
\UseRawInputEncoding
\usepackage{hyperref}
\usepackage{bm}
\usepackage{mathrsfs}
\usepackage{amsmath}
\usepackage{amsthm}

\newcommand{\tabincell}[2]{\begin{tabular}{@{}#1@{}}#2\end{tabular}}
\usepackage{paralist}
\usepackage{multirow}
\usepackage{graphicx}
\usepackage{amsfonts}
\usepackage{amsmath,epsfig}
\usepackage{stfloats}
\usepackage{amssymb}
\usepackage{epstopdf}
\usepackage{latexsym}
\usepackage{color}
\usepackage{cite}
\usepackage{algorithm}
\usepackage{algorithmic}

\ifCLASSINFOpdf

\else

\fi

\begin{document}
%
\title{Task-Oriented Delay-Aware Multi-Tier Computing in Cell-free Massive MIMO Systems}

\author{Kunlun Wang, \IEEEmembership{Member,~IEEE}, Dusit Niyato, \IEEEmembership{Fellow,~IEEE}, Wen Chen, \IEEEmembership{Senior~Member,~IEEE}, and Arumugam Nallanathan \IEEEmembership{Fellow,~IEEE}

\thanks{This work was supported by National key project 2020YFB1807700. (Corresponding author: Kunlun Wang.)} 

\thanks{K. Wang is with the Shanghai Key Laboratory of Multidimensional Information Processing, East China Normal University, Shanghai 200241, China, and also with the School of Communication and Electronic Engineering, East China Normal University, Shanghai 200241, China (e-mail: klwang@cee.ecnu.edu.cn).

D. Niyato is with School of Computer Science and Engineering, Nanyang Technological University, Singapore (e-mail: dniyato@ntu.edu.sg).

W. Chen is with the Department of Electronic Engineering, Shanghai Jiao Tong University, Shanghai 200240, China (e-mail: wenchen@sjtu.edu.cn).

A. Nallanathan is with the School of Electronic Engineering and Computer Science at Queen Mary University of London, UK. (email: a.nallanathan@qmul.ac.uk).}
}
\maketitle





\begin{abstract}
Multi-tier computing can enhance the task computation by multi-tier computing nodes. In this paper, we propose a cell-free massive multiple-input multiple-output (MIMO) aided computing system by deploying multi-tier computing nodes to improve the computation performance. At first, we investigate the computational latency and the total energy consumption for task computation, regarded as total cost. 
Then, we formulate a total cost minimization problem to design the bandwidth allocation and task allocation, while considering realistic heterogenous delay requirements of the computational tasks. Due to the binary task allocation variable, the formulated optimization problem is non-convex. Therefore, we solve the bandwidth allocation and task allocation problem by decoupling the original optimization problem into bandwidth allocation and task allocation subproblems. As the bandwidth allocation problem is a convex optimization problem, we first determine the bandwidth allocation for given task allocation strategy, followed by conceiving the traditional convex optimization strategy to obtain the bandwidth allocation solution. Based on the asymptotic property of received signal-to-interference-plus-noise ratio (SINR) under the cell-free massive MIMO setting and bandwidth allocation solution, we formulate a dual problem to solve the task allocation subproblem by relaxing the binary constraint with Lagrange partial relaxation for heterogenous task delay requirements. At last, simulation results are provided to demonstrate that our proposed task offloading
scheme performs better than the benchmark schemes, where the minimum-cost optimal offloading strategy for heterogeneous delay requirements of the computational tasks may be controlled by the asymptotic property of the received SINR in our proposed cell-free massive MIMO-aided multi-tier computing systems.

\end{abstract}
\IEEEpeerreviewmaketitle
\section{Introduction}
\IEEEPARstart{S}{ince} there are exponential growth of mobile devices in the networks, wireless traffic is growing tremendously recently. It is expected that and the increasing amount of traffic will continue to grow steadily in the coming years. With the proliferation of wireless traffic, delay and energy consumption have emerged as key design metrics for wireless communication systems~\cite{2015-Wkl-packet,2018-She-URLLC}. Additionally, more and more smart devices are connected to the wireless network with the development of intelligent Internet of Things~(IoT), it is estimated that in excess of $24.6$ billion connected devices by $2025$~\cite{2017-Shikhar-IOT}. Meanwhile, the significant growth of novel intelligent applications with intensive tasks (e.g., AR/VR) typically require ultra-reliable and low-latency communications (URLLC) and demand efficient power management for realtime task processing and high energy efficiency (EE)~\cite{2019-She-URLLC,2022-Dang-URLLC}. However, mobile hand-held devices have limited computation, energy and storage resources, as well as limited battery capacity due to their compact form-factor.
These defects pose critical challenges for the realtime intelligent applications. By enabling flexible computation, storage and communication resource coordination, multi-tier computing is a novel and efficient computing architecture, which can schedule intensive tasks to multi-tier computing servers at heterogeneous base stations (BSs) in the edge/fog or cloud of wireless communication systems \cite{2020-Wang-NOMAFOG,2018-Yang-MEETS}. 

From the perspective of EE, the energy consumption model adopted in most of related works is over-simplified, which only models the transmission energy or task computational energy \cite{2016-Chen-Offloading,2018-Yang-MEETS}. However, a more general energy model is not negligible in a multi-tier computing system considering various computation and communication energy consumptions at the multi-tier nodes.
From the perspective of task computation delay, the computation delay model proposed in most of related works fails to capture the effects of heterogeneous delay requirements.
There are diversified applications in multi-tier computing systems, some tasks are delay tolerant, while some tasks are delay sensitive. However, most of these works fail to consider \textcolor{black}{joint influence} of the heterogeneous delay requirements and the total energy consumption, which can be also regarded one of the key task computation metrics in next generation wireless networks.

\textcolor{black}{Regarded as one of the key technologies for next generation wireless communication systems, massive multiple-input multiple-output (MIMO) is capable of significantly improving the task offloading rate so as to improve the task compuation efficiency \cite{2020-Wang-MMIMO,2022-Wang-MMIMOCache,2022-Wang-MMIMOIRS}. In order to make full use of the benefits of massive MIMO in a multi-tier computing system, the cell-free massive MIMO based multi-tier computing is considered. Cell-free massive MIMO is the network-centric massive MIMO, which is distributed across the network \cite{2017-Trung-CMIMO,2019-Manijeh-CFMMIMO}. Thus, a cell-free massive MIMO system supports multiple number of antennas distributed over a large number of access points (APs) in a network, where each AP can serve a small number of devices. These APs are coordinated by the central processing unit (CPU) through a high-transmission data rate and ultra-reliable backhaul link.}
For implementing multi-tier computing, we assume that each AP and the CPU are respectively equipped with an independent computing server, while the computing server for CPU has much larger computation capacity compared to that of AP. Each AP serves all devices in its coverage area. Thus, each device is capable of utilizing the abundant computational resources on the CPU or one of its connected APs (i.e., multi-tier computing framework).
In this context, our proposed computing framework focuses on multi-tier computing nodes (CNs) in cell-free massive MIMO systems, e.g., fog access nodes (FANs), cloud access nodes (CANs) and CPU, realizing collaborative task computing. These CNs are able to help to execute the computational tasks offloaded from task nodes (TNs) according to their computing capabilities, task delay requirements and total cost.

\subsection{Related Work}
In order to address the inefficient task computation issues, multi-tier computing is capable of offloading the intensive tasks to multi-tier nearby CNs with the powerful computing capability realizing remote task execution \cite{2022-Wang-MC}. Therefore, task offloading has received more and more research attention in different edge/fog computing scenarios~\cite{2018-Yang-MEETS,2019-Wang-TVT,2020-Wang-NOMAFOG,2020-Wang-MMIMO,2020-Bai-MECIRS,2022-Dang-MEC,2022-Wang-MMIMOCache,2022-Wang-MMIMOIRS}.
In particular, Wang \MakeLowercase{\textit{et al.}} \cite{2020-Wang-MMIMO} proposed a massive MIMO-aided task offloading system, where multiple TNs can schedule their tasks to nearby computing nodes (CNs) by a massive MIMO-based FAN. As an extension of this work, Wang \MakeLowercase{\textit{et al.}} \cite{2022-Wang-MMIMOCache} proposed a relay assisted multi-tier computing system equipped with massive antennas to improve the task computation performance, they investigated the design of the task allocation, service caching and power allocation jointly to minimize the total task computation latency.
Additionally, Wang \MakeLowercase{\textit{et al.}} \cite{2020-Wang-NOMAFOG} investigated a non-orthogonal multiple access (NOMA)-assisted task offloading system for industrial Internet of things (IIoT), where TNs offload their intensive tasks to nearby CNs via NOMA technique for task computation. Also Wang \MakeLowercase{\textit{et al.}} \cite{2022-Wang-MMIMOIRS} proposed a task offloading framework in a intelligent reflecting surface (IRS) and massive MIMO relay assisted multi-tier computing system, where multiple TNs can offload their intensive tasks to nearby massive MIMO relay node~(MRN) and FAN via the IRS technique for task execution in CNs.
Liu \MakeLowercase{\textit{et al.}} \cite{2019-Yang-Pomt,2020-Liu-Post} proposed a mapping framework for task offloading by mapping multiple tasks or TNs into multiple HNs, where they studied a generalized Nash equilibrium problem to minimize task's offloading delay in a distributed manner.

By equipping with a very large number of distributed APs in a network, cell-free massive MIMO is capable of significantly improving network-throughput as well as EE \cite{2017-Nayebi-CFMMIMO}. Given the benefits of cell-free massive MIMO, the integration of multi-tier computing and cell-free massive MIMO can improve the task offloading performance in a multi-tier computing system~\cite{2017-Nayebi-CFMMIMO,2020-Sudarshan-cellMIMO,2021-Ke-CFMMIMO}. There are some works on
resource management of task offloading in a cell-free massive MIMO system in recent years.
Mukherjee \MakeLowercase{\textit{et al.}} \cite{2020-Sudarshan-cellMIMO} proposed an edge computing-assisted cell-free massive MIMO architecture, where the edge servers and cloud are located at each AP and the central server of this system, respectively, and the authors analysis the task offloading performance by devising suitable communication resource allocation and task allocation strategies. Ke \MakeLowercase{\textit{et al.}} \cite{2021-Ke-CFMMIMO} introduced a grant-free massive access IoT system, where multiple cooperative APs serve massive devices in the network via cell-free massive MIMO technique, they studied two computation strategies at the CNs for massive devices access, i.e., cloud task computation and edge task computation. Wang \MakeLowercase{\textit{et al.}} \cite{2022-Wang-MMIMOCache} investigated the joint strategy of task offloading, computational task caching and power allocation in edge computing systems to minimize the total task offloading latency.

Although the above works have revealed the benefits of cell-free massive MIMO-assisted edge computing \cite{2020-Wang-MMIMO,2022-Wang-MMIMOCache,2022-Wang-MMIMOIRS}, the energy-delay tradeoff in resource management and multi-tier task computation have not been considered.
By considering task computation energy and latency costs in a multi-tier task offloading framework, scheduling tasks to multi-tier CNs can reduce the congestion of task computation as well as reducing the computation energy consumption of each user. To exploit the benefits of energy-delay tradeoff in edge computing framework, there are some works being invested into online dynamic tasks allocation with energy harvesting \cite{2018-Zhang-ERTRA}, task offloading of mobile devices formulated as a constrained multi-objective optimization problem on minimizing both the task computation energy consumption and task computation latency \cite{2021-Arash-MOC}, and jointly task offloading and resource management optimization \cite{2022-Li-EEMEC}, as well as minimizing task response time and packet losses to improve the realtime performance and reliability of task processing \cite{2019-Deng-POG}. \textcolor{black}{However, the influence of heterogeneous delay requirements of computational tasks has not been studied, which represents the different delay requirements of diverse novel applications. Additionally, the influence of heterogeneous delay requirements of the tasks for task allocation in multi-tier computing nodes has not been studied either. Furthermore, all the existing works consider the uncoordinated distributed edge computing scenario.} Thanks to the rapid development of cell-free massive MIMO \cite{2015-Ngo-CFMIMO}, the task offloading via cell-free massive MIMO will be increasingly adopted in multi-tier computing framework.

\begin{table*}[!t]\small
\caption{Novelty Comparison}\label{tablec}
\centering
\begin{tabular}{|l|c|c|c|c|c|c|c|c|c|} 
\hline
 & \tabincell{c}{\cite{2016-Chen-Offloading}-\\2016} & \tabincell{c}{\cite{2019-Wang-TVT}-\\2019} & \tabincell{c}{\cite{2017-Mao-stochastic}-\\2017} & \tabincell{c}{\cite{2019-Ozgun-FogMIMO}-\\2019} & \tabincell{c}{\cite{2020-Wang-NOMAFOG}-\\2020} & \tabincell{c}{\cite{2020-Wang-MMIMO}-\\2020} & \tabincell{c}{\cite{2022-Wang-MMIMOCache}-\\2022} & \tabincell{c}{\cite{2022-Wang-MMIMOIRS}-\\2022} &  \tabincell{c}{Our \\ work} \\
\hline
\tabincell{c}{Joint task allocation and communication\\resource allocation} & $\checkmark$ & $\checkmark$ &  $\checkmark$ &  & $\checkmark$ & $\checkmark$ & $\checkmark$ & $\checkmark$ & $\checkmark$\\
\hline
Massive MIMO &   &  &  & $\checkmark$ &  & $\checkmark$ & $\checkmark$ & $\checkmark$ & $\checkmark$\\
\hline
\tabincell{c}{Minimizing task computation time}  & $\checkmark$  &  &   &  &  & & & & $\checkmark$\\
\hline
\tabincell{c}{Multi-tier computing} &   &  &   &  & $\checkmark$  & $\checkmark$  &$\checkmark$  & $\checkmark$  & $\checkmark$\\
\hline
\tabincell{c}{Cell-free massive MIMO} &   &  &   &  &  &  & & & $\checkmark$\\
\hline
\tabincell{c}{Energy-delay tradeoff} &   &  &   &  &  &  & & & $\checkmark$\\
\hline
\tabincell{c}{Heterogeneous delay requirements} &   &  &   &  &  &  & & & $\checkmark$\\
\hline
\end{tabular}
\label{table1} 
\end{table*}

\subsection{Main Contributions}

Although the above contributions have revealed the benefits of task offloading in a cell-free massive MIMO framework, the multi-tier collaborative task computation minimizing the total energy consumption and latency in cell-free massive MIMO frameworks has not been considered to the best of our knowledge. The influence of delay and energy weight has not been well studied either, which can be regarded as the heterogeneous delay requirements of the tasks and will be a key performance metric for multi-tier computing systems~\cite{2016-Wang-delay,2018-Zhao-FEMOS,2018-Lyu-EEdelay,2020-Wang-NOMAFOG}, 
and the heterogeneous delay requirements is particularly important for battery-limited mobile devices running diverse novel applications. The proposed cell-free massive MIMO-aided multi-tier computing framework includes heterogeneous CNs, e.g., mobile devices, fog/cloud access points, and cloud. The total energy consumption consist of both the task transmission energy and task computation energy, and the total latency cost includes both the task transmission delay and task computation delay. The task computation energy includes the energy consumed by both the local devices and the remote CN.
Furthermore, to characterize the effects of different delay requirements for the tasks in multi-tier computing systems, we employ the delay and energy weight for determining the task allocation policy. That is, delay sensitive and tolerant tasks are classified into different categories by the weights. Based on the requirements of delay sensitive and tolerant tasks, we aim for studying how cell-free massive MIMO systems tackle the challenge of task allocation. \textcolor{black}{In the expressions for the task computation latency and energy consumption, we obtain that the bandwidth allocation variable and task offloading variable are separated. Then, the task allocation and the bandwidth allocation constraints are separable. Therefore, we decouple the original task offloading Problem into two sub-problems: bandwidth allocation optimization and task allocation optimization.}
Specifically, we first exploit the optimal bandwidth allocation strategy to minimize the total cost of our proposed cell-free massive MIMO-assisted multi-tier computing systems. Secondly, we characterize the relationship between cell-free massive MIMO and task allocation strategy based on the asymptotic property of the signal-to-interference-plus-noise ratio (SINR) with very large number of APs. Finally, we optimize the task allocation for heterogeneous delay requirements of the tasks.
We have explicitly compared our unique contributions with the state-of-the-art, which are shown in Table \ref{tablec} and further summarized as follows:
\begin{itemize}
\item \textcolor{black}{In terms of multi-tier task offloading system model, a novel cell-free massive MIMO assisted multi-tier computing framework is proposed for task offloading, which consists of a group of a large number of distributed APs with sharable computing resources to execute the tasks. Furthermore, a problem minimizing the total energy consumption and latency cost is formulation, which offers a new approach to task offloading in a multi-tier computing system.} 
\item \textcolor{black}{In terms of task offloading optimization and resource management analysis, an adaptive bandwidth allocation strategy is proposed for minimizing the total energy and delay cost, where the most efficient bandwidth for each task offloading link can be determined according to the dynamic channels. As far as we know, considering the bandwidth allocation in a cell-free massive MIMO-assisted task offloading system has not been studied recently.}
\item Furthermore, we optimize the task allocation strategy based on the bandwidth allocation result. Given the delay sensitive and tolerant tasks, we study how cell-free massive MIMO systems tackle the challenge of task allocation with heterogeneous delay requirements. Then, we employ the delay and energy weight to determine the optimal task allocation policy.
\item  Finally, the performances of the proposed bandwidth allocation and task offloading strategy have been evaluated through simulations relying on diverse system parameters. The simulation results demonstrate that our task offloading strategy achieves significantly performance improvement in total cost compared to the benchmark schemes subject to realistic communication and computation constraints.
\end{itemize}

\subsection{Paper Organization}
Our paper is organized as follows. The system model is introduced in Section \uppercase\expandafter{\romannumeral2}, while the problem formulation and analysis are presented in Section \uppercase\expandafter{\romannumeral3}. In Section \uppercase\expandafter{\romannumeral4}, we optimize the bandwidth and task allocations in terms of heterogeneous delay requirements by the asymptotic property of the massive MIMO to minimize the total energy and delay cost in cell-free massive MIMO-assisted multi-tier computing systems.
our simulation results \textcolor{black}{are shown in Section} \uppercase\expandafter{\romannumeral5}. Finally, conclusions are provided in Section \uppercase\expandafter{\romannumeral6}. Table \ref{table2} lists the notations.


\begin{table*}[!t]
\caption{Notations}\label{table2}
\centering
\textcolor{black}{\begin{tabular}{lc|lc} 
\hline
Definition & Notation & Definition & Notation  \\
\hline
Bandwidth & $B$ & Task size of task $k$ & $l_k$ \\
Bandwidth allocation variable for TN $k$ & $\eta_{k}$ & Task offloading decision variable of task node $k$ & $\alpha_k$ \\
Transmit power of TN $k$ & $p_b$ & Variance of AWGN & $\sigma_k^2$ \\
Transmit symbol of task node $k$ & $s_k$ & Received SINR of symbol $x_k$ & $\gamma_{k}$\\
Task offloading time of first hop & $t_k^{\rm{TN}}$ & Offloading energy of first hop & $E_k^{\rm{TN}}$\\
Task offloading time of second hop & $t_k$ & Offloading energy of first hop & $E_k^{\rm{FAN}}$\\
Computational latency at CN & $t^F_{k,\rm{comp}}$ & Computational latency at CPU & $t^C_{k,\rm{comp}}$\\
Computational energy consumption at CN & $E_{\rm{re},k}$ & Total energy consumption of TN $k$ & $E_{\rm{total},k}$\\
Total task offloading latency of task node $k$ & $T_{\rm{total},k}$ & Transmitted symbol from the $k$th FAN & $x_{k}$ \\
Weights of delay and energy consumption & $\mu_{k}$ & Weighted sum cost of task node $k$ & $\Omega_{k}$\\
Total cost & $\Omega_{\rm{total}}$ & Number of CAN & $M$ \\
\hline
\end{tabular}}
\label{table1} 
\end{table*}

\begin{figure}[!t]
\begin{center}
\includegraphics [width=3.7in]{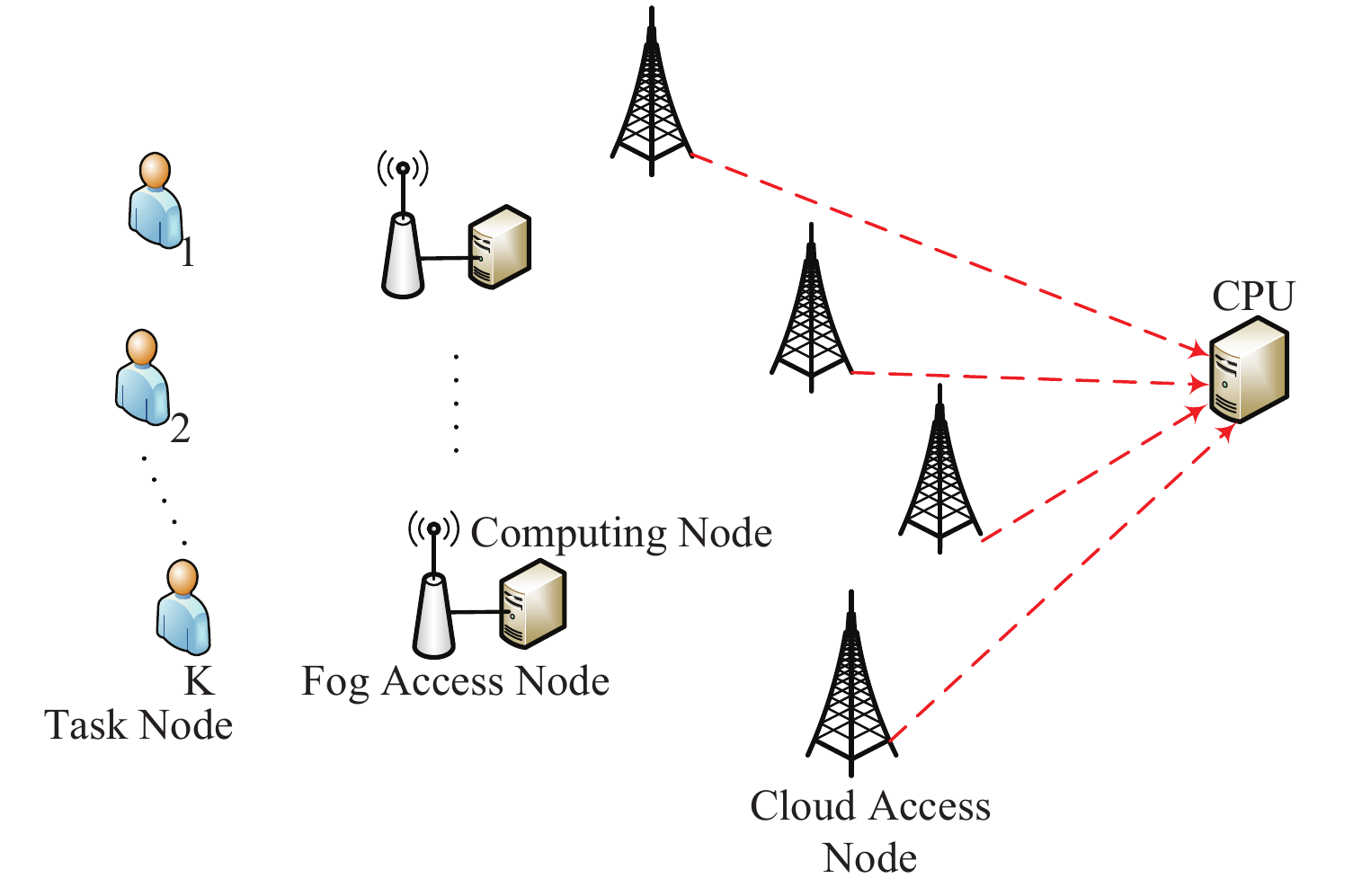}
\caption{Illustration of a cell-free massive MIMO-assisted multi-tier computing system consisting of $K$ TNs, $K$ fog access nodes (FANs) and $M$ cloud access nodes (CANs), where TNs offload the tasks to the FANs or CANs for multi-tier computing.} \label{Fig-system_model}
\end{center}
\end{figure}

\section{System Model}
\subsection{Multi-tier Computing Network}
\textcolor{black}{A cell-free massive MIMO-aided multi-tier computing network is shown in Fig.~\ref{Fig-system_model}, which consists of $K$ TNs with computational tasks (e.g., artificial intelligence model training), $K$ FANs and $M$ CANs, where the active TNs rely on node to node (N2N) communications for task offloading to FANs. The FANs are connected to the CNs by backhaul links using high-rate ultra-reliable optical fiber transmission.} The TNs are communicated with the FAN by wireless links, and the FANs offload the tasks to the CANs, while the CANs are connected to the CPU by backhaul links using high-rate optical fiber. As the next generation wireless network is expected to satisfy the extensive quality of service (QoS) requirements for wireless communications, a cell-free massive MIMO-based multi-tier computing architecture is proposed. Multi-tier computing organizes and manages the task computation and transmission from multiple heterogeneous CNs. Under this circumstance, the tasks can be offloaded to the FAN and CPU with the help of CANs, which can improve the performance of the overall computation efficiency. 

\subsection{Transmission Model}
\textcolor{black}{We consider independent and identically distributed (i.i.d.) quasi-static Rayleigh fading. In particular, each inter-node channel remains invariant within one time slot, but varies independently across different time slots and links.}
\textcolor{black}{To make full use of the spectrum, we consider that each TN occupies a part of the bandwidth to facilitate task offloading by frequency division multiple access (FDMA), which can avoid the TNs¡¯ co-channel interference.} Let $B$ and $\eta_{k}\in[0,1]$ be the total available bandwidth for the links between TNs and FANs and the bandwidth allocation variable for TN $k$, respectively, recall that the bandwidth allocated to the $i$th TN is $\eta_k B$. 
We denote $\bm{\eta}$ as the bandwidth allocation vector, which can be shown as $\bm{\eta}=[\eta_{1},\eta_{2},\cdots,\eta_{K}]$.

We denote $\alpha_{k}\in\{0,1\}$ and $\mathbf{a}=[\alpha_{1},\cdots,\alpha_{K}]$ as the offloading variable of TN $k$ and the offloading vector, respectively. As we consider binary task offloading, we have $\alpha_{k}=0$ if TN $k$ executes its task locally by FAN computing, and $\alpha_{k}=1$ if TN $k$ executes its task by remote CPU computing. \textcolor{black}{Let $p_b$ be the transmit power of each TN, which is selected prior to the subcarrier allocation.} Denote $z_k$ as the additive white Gaussian noise~(AWGN) with variance $\sigma_k^2$ at the $k$th FAN. Hence, the received task signal at the $k$th FAN is shown as
\begin{equation}\label{receives}
y_k=\sqrt{p_b}\sqrt{d_{k}}h_{k}x_k+z_k,\,\forall k\in\mathcal{K},
\end{equation}
where $d_{k}$ and $h_{k}\sim\mathcal{CN}(0,1)$ represent the channel path loss and small scale fading of the link between the $k$th TN and the $k$th FAN. Let $x_k$ be the transmit symbol of TN $k$. In addition, $x_k$ satisfies $\textsf{E}[|x_k|^2]=1$. 

Based on \eqref{receives}, the \textcolor{black}{signal-to-noise ratio (SNR)} of symbol $x_k$ observed by the $k$th FAN can be expressed as
\begin{equation}\label{fsinr}
\gamma_{\rm{FAN},k}=\frac{p_{b}d_{k}|h_{k}|^2}{\sigma_k^2}.
\end{equation}
Then, based on the bandwidth allocation and received SNR in \eqref{fsinr}, the achievable task offloading rate of the first hop from the $k$th TN to the $k$th FAN is given by
\begin{equation}\label{adrate}
r_{k}=\eta_kB\log_2\left(1+\frac{p_{b}d_{k}|h_{k}|^2}{\sigma_k^2}\right).
\end{equation}
Then, the task offloading time of the first hop is given by
\begin{equation}\label{firstdelay}
t^{\rm{TN}}_{k}=\frac{l_k}{\eta_kB\log_2\left(1+\frac{p_{b}d_{k}|h_{k}|^2}{\sigma_k^2}\right)}.
\end{equation}
The corresponding energy consumption is given by
\begin{equation}\label{firstenergy}
E^{\rm{TN}}_{k}=\frac{p_{b}l_k}{\eta_kB\log_2\left(1+\frac{p_{b}d_{k}|h_{k}|^2}{\sigma_k^2}\right)}.
\end{equation}

For the second hop, FAN $k$ offloads its task to the CAN. By the cell-free transmissions, CAN receives all the tasks from the FANs. Then, the task offloading time and task offloading energy consumption of the second hop are respectively given by
\begin{equation}\label{seconddelay}
t_{k}=\frac{\alpha_kl_k}{B\log_2\left(1+\gamma_k\right)},
\end{equation}
\begin{equation}\label{seconddelay2}
E^{\rm{FAN}}_{k}=\frac{q_k\alpha_kl_k}{B\log_2\left(1+\gamma_k\right)},
\end{equation}
where $\gamma_k$ represents the received SINR for the TN $k$ at CANs.

\subsection{Computational Model}
\subsubsection{Local Computing}
\textcolor{black}{Consider that each TN has a computational task for the requested service, and hence we denote $k_s$ as the task of TN $k$,} which can be specified by the task size with $l_{k}$ bits. \textcolor{black}{For simplicity of analysis, we assume that the task for each TN is generated instantly, and the task offloading latency is from the arrival in the TN. Regrading task computation in TN, we assume that TN does not have enough computation capacity to execute the task, the task need to be offloaded to execute on nearby FAN. If the task is still computed locally, the execution delay could be largely due to the limited computing resources at the TN \cite{2016-Shi-Edgecomp}.} 

\textcolor{black}{In our proposed multi-tier computing system, the total number of computing cycles of processing core is considered to be linearly proportional with the task size of each task to be processed \cite{2020-Wang-NOMAFOG,2020-Wang-MMIMO}}. \textcolor{black}{Let $C_{k}$ represent the number of computing cycles required for executing $1$-bit of input task at CN, and hence the total number of computing cycles required for executing the task from TN $k$ is $C_kl_{k}$, which depends both on the type of processing core and on the task to be executed.}
Therefore, the task computational latency for $k_s$ at CN is given by
\begin{equation}\label{t1}
t_{k,\rm{comp}}^F=\frac{C_k(1-\alpha_{k})l_{k}}{f_k^F},
\end{equation}
where $f_k^F$ represents the core computing frequency at the CN.
\subsubsection{Task Offloading}
In terms of task offloading, where the task is offloaded to be processed by remote CPU. In this case, the TN transmits the task to the CPU through two hops wireless links, i.e, TNs to FANs link and FANs to CANs link. \textcolor{black}{For the sake of simplicity, we assume that the CPU has multi-core, and each particular offloaded task can be independently assigned to a core Furthermore, each core is assumed to have} the same maximum computing frequency $f^{\max}_0$ (in cycles per second).
Therefore, we have computational latency of the task $k_s$ at CPU, which is given by
\begin{equation}\label{t2}
t_{k,\rm{comp}}^{C}=\frac{C_k\alpha_{k}l_{k}}{f_k^C},
\end{equation}
where $f_k^C$, $k\in\mathcal{K}$ represents the each core computing frequency in the CPU. \textcolor{black}{Benefit from} dynamic voltage and frequency scaling techniques (DVFS) \cite{2017-Mao-MEC}, we assume that $f_k^C$ is adjustable.

\section{Problem Formulation and Analysis}
\subsection{Energy Consumption and Delay Analysis}
Following the model in~\cite{2016-Chen-Offloading}, let $P_{\mathrm{CN}}$ denote the computing energy consumption for each cycle at CN for local task computing. Then, $C_{k}P_{\mathrm{CN}}$ represents the computing energy consumption per bit of each task. According to the task allocation decision, \textcolor{black}{there are $(1-\alpha_{k})l_{k}$-bits input task required to be processed at $k$th CN.} Thus, the computational energy consumption of the task $k_s$ at the $k$-th CN for local task computing is given by
\begin{equation}\label{Cenergy}
E_{\mathrm{re},k}=C_k(1-\alpha_{k})l_{k}P_{\mathrm{CN}}.
\end{equation}
Specifically, the computational energy consumption of the task $k_s$ at CPU is not considered, as the energy capacity of CPU is very large.
Then, the total computational energy consumption for task $k_s$ is given by
\begin{equation}\label{totalcenergyk}
E_{\mathrm{com},k}=E_{\mathrm{re},k}=C_0(1-\alpha_{k})l_{k}P_{\mathrm{CN}}.
\end{equation}
In all, the total task computation energy consumption for task $k_s$ is composed of total task transmission energy and total task computation energy, which is given by
\begin{equation}\label{totalenergyk}
\begin{aligned}
E_{\mathrm{total},k}=&E^{\rm{TN}}_{k}+E_{\mathrm{re},k}+E^{\rm{FAN}}_{k}\\
=&\frac{p_{b}l_k}{\eta_kB\log_2\left(1+\frac{p_{b,m}d_{b,m}|h_{b,m}|^2}{\sigma_m^2}\right)}+\\
&C_0(1-\alpha_{k})l_{k}P_{\mathrm{CN}}
+\frac{q_k\alpha_kl_k}{B\log_2\left(1+\gamma_k\right)}.
\end{aligned}
\end{equation}

\textcolor{black}{The total task offloading latency consists of the task computation delay plus the task transmission delay, and it is from the arrival in the TN. According to \eqref{firstdelay}, \eqref{seconddelay}, \eqref{t1}, and \eqref{t2}, the total task offloading latency of task $k_s$ is given by}
\begin{equation}\label{delayk}
\begin{aligned}
T_{\rm{total},k}=&t^{\rm{TN}}_{k}+t_{k,\rm{comp}}^F+t_{k}+t_{k,\rm{comp}}^{C}\\
=&\frac{l_k}{\eta_kB\log_2\left(1+\frac{p_{b}d_{b,m}|h_{b,m}|^2}{\sigma_m^2}\right)}+\frac{C_k(1-\alpha_{k})l_{k}}{f_k^F}\\
&+\frac{\alpha_kl_k}{B\log_2\left(1+\gamma_k\right)}+\frac{C_k\alpha_{k}l_{k}}{f_k^C}.
\end{aligned}
\end{equation}

\subsection{Problem Formulation}
In this section, the bandwidth and task allocations framework is discussed for heterogeneous delay requirements of the tasks under the cell-free massive MIMO setting, and the related optimization problem is formulated. Specifically, the total task computation cost in terms of latency cost and energy consumption minimization is formulated to improve the task execution efficiency of multi-tier task computation.

According to \eqref{totalenergyk} and \eqref{delayk}, the weighted sum cost for each TN $k$ in terms of latency cost and energy consumption can be obtained, which is given by
\begin{equation}
\Omega_{k}=\mu_kE_{\mathrm{total},k}+(1-\mu_k)T_{\rm{total},k},
\end{equation}
where $\mu_k$ and $(1-\mu_k)$ represent the weight of latency cost and energy consumption, respectively, and $0\leq\mu_k\leq 1$. $\mu_k$ reflects the priority of different task requirements, under the condition of delay sensitive task, $\mu_k$ is small. On the other hand, if the task is delay tolerant, $\mu_k$ is large, which means total energy consumption is more important than delay performance for task computation. For simplicity of analysis, we set $\mu_k=0$ for delay sensitive services, and we set $\mu_k=1$ for delay tolerant services.
Then, the total cost consists of computational energy consumption and computational delay, which can be formulated as
\begin{equation}
\begin{aligned}
\Omega_{\rm{total}}=&\sum_{k=1}^Kw_k\Omega_{k}\\
=&\sum_{k=1}^Kw_k\left\{\mu_kE_{\mathrm{total},k}+(1-\mu_k)T_{\rm{total},k}\right\},
\end{aligned}
\end{equation}
where $w_k$ is the weighting factor for TN $k$.

Hence, the total cost minimization optimized problem of the proposed cell-free massive MIMO asisted multi-tier computing systems is expressed as
\begin{subequations}\label{originalP}
\begin{align}
\min\limits_{\bm{\eta},\bm{\alpha}} \quad  &\Omega_{\rm{total}},\\
\rm{s.t.}\quad     
      &0 < \eta_k< 1, \forall k,\\
       &\sum_{k=1}^K\eta_k= 1, \forall k,\\
         \label{five}  & \alpha_{k}\in  \{0,1\},\forall k.
 \end{align}
\end{subequations}

\subsection{Problem Analysis}
\textcolor{black}{Using the above described problem formulation, we need to analyze Problem \eqref{originalP}. Note that $\alpha_{k}$ is a binary variable, and hence Problem \eqref{originalP} has an non-convex feasible set. As binary variables has product-based relationship, the objective function of Problem \eqref{originalP} is also non-convex. As we know that it is challenging to find global optimum for a mixed discrete and non-convex optimization problem \cite{2013-Xiao-QoS}. Since the feasible set and objective function are both non-convex, this problem is NP hard. Thus, we have to simplify Problem \eqref{originalP}.}

From \eqref{totalenergyk} and \eqref{delayk}, it can be observed that in the expressions for the task computation latency and energy consumption, the bandwidth allocation variable $\bm{\eta}$ and task offloading variable $\bm{\alpha}$ are separated. In this case, the task allocation and the bandwidth allocation constraints are separable. Therefore, we are capable of decoupling the original task offloading Problem \eqref{originalP} into two sub-problems: bandwidth allocation optimization (i.e., $\bm{\eta}$) and task allocation optimization (i.e., $\bm{\alpha}$). Note that the bandwidth allocation is based on the first hop task transmission from the TNs to the FANs, and the task allocation is based on the second hop task offloading from the FANs to CANs with the aid of cell-free massive MIMO.

\section{Cost Minimization Schemes}
In this section, the bandwidth allocation and task offloading is optimized for heterogeneous delay requirements of the tasks by the asymptotic property of the massive MIMO to minimize the total task computation cost in our proposed cell-free massive MIMO-aided multi-tier computing systems.

\subsection{Uplink Task Transmission}
Consider that task transmissions from FANs to CANs are regarded as uplink task offloading, where all FANs send their received tasks to the CANs. We denote the channel fading coefficient between the FAN $k$ and the CAN $m$ as $g_{m,k}$, which can be expressed as \cite{2017-Ngo-CFMMIMO}
\begin{equation}
g_{m,k}=\sqrt{\beta_{m,k}}h_{m,k},
\end{equation}
where $\beta_{m,k}$ denotes the large-scale fading and $h_{m,k}\sim\mathcal{CN}(0,1)$ denotes the small-scale channel fading between the link from the $k$th TN and the FAN $m$. Upon receiving the signal, all the FANs delivered their symbols simultaneously to the CAN, which can be expressed as 
\begin{equation}
\mathbf{x}=\sqrt{\rho q_k}\mathbf{s},
\end{equation}
where $\mathbf{s}=[s_1,\cdots,s_K]^T$ denotes the transmission information-bearing signal vector with $\textsf{E}(\mathbf{s}\mathbf{s}^\dag)=\mathbf{I}_K$, and $s_k$ ($\textsf{E}[|s_k|^2]=1$) denotes the signal transmitted from the $k$th TN to its paired FAN, and $q_k$ denotes the task transmission power from the FAN $k$. Additionally, $\rho$ represents the normalized uplink signal-to-noise ratio (SNR).
The delivered symbol from FAN $k$ can be expressed as
\begin{equation}
x_k=\sqrt{\rho q_k}s_k.
\end{equation}

Note that each CAN receives aggregated signal from all the FANs, and the received symbol at the $m$th CAN can be expressed as
\begin{equation}
y_m=\sqrt{\rho}\sum_{k=1}^Kg_{m,k}\sqrt{q_k}s_k+n_m,
\mathbf{y}_m=\mathbf{g}_{m,k}\sqrt{\rho q_k}\mathbf{s}+\mathbf{n}_m,
\end{equation}
where $n_{m}\sim\mathcal{CN}(0,1)$ denotes the noise at CAN $m$. Additionally, we employ a matched filtering strategy at the CANs. In this case, the received symbol can be weighted appropriately. To be more precisely, the received symbol $y_m$ at the $m$th CAN need to be first multiplied by $\hat{\mathbf{g}}^\ast_{m,k}$.
Let $\mathbf{\hat{g}}_m$ denote the estimated channel state information (CSI) of all FANs to $m$th CAN. Under this circumstance, the actual channel coefficient can be expressed
as~\cite{2011-Andrews-MIMO}
\begin{equation}
\mathbf{g}_m=\sqrt{1-\tau_D^2}\mathbf{\hat{g}}_m+\tau_D\mathbf{\Omega}_D,
\end{equation}
where $\mathbf{\Omega}_D\in\mathbb{C}^{K\times M}$ represents the channel estimation noise independent of the estimated channel, which has i.i.d entries with zero mean and unit variance, and $\tau_D\in [0,1]$ reflects the estimation accuracy of the channel coefficient matrix $\mathbf{\hat{g}}_m$. Under the condition of $\tau_D=0$, we have perfect CSI estimation. Under the condition of $\tau_D=1$, the CSI is completely unknown. \textcolor{black}{Due to the estimation errors, limited CSI feedback quantization, delays, etc. On the other hand, the Quality of Experience (QoE) of computation heavily relies on the wireless fading channel conditions since task offloading requires effective wireless transmission. Thus, we consider imperfect CSI for task offloading.}
Given the imperfect CSI of receiver, the precoding matrix of the CAN is similar to \cite{2020-Wang-MMIMO,2022-Wang-MMIMOCache}, which can be expressed as
\begin{equation}
\hat{\mathbf{g}}^\ast_{m}=\mathbf{\hat{g}}_{m}^{\dagger},\label{detector1}
\end{equation}
where $\mathbf{\hat{g}}_{m}^{\dagger}=(\mathbf{\hat{g}}_{m}^\mathrm{H}\mathbf{\hat{g}}_{m})^{-1}\mathbf{\hat{g}}_{m}^\mathrm{H}$.

The multiplied symbol of $\hat{\mathbf{g}}^\ast_{m,k}\mathbf{y}_m$ is then delivered to CPU via backhaul link using high-throughput optical fiber for signal detection. 
Then, the received data at CPU is aggregated, which can be expressed as
\begin{equation}\label{CPUrece}
\begin{aligned}
\mathbf{r}_{\mathrm{CPU}}&=\sum_{m=1}^M\hat{\mathbf{g}}^\ast_{m}\mathbf{y}_m\\
=&\sqrt{\rho q_k}\sum_{m=1}^M\mathbf{g}_{m}^{\dagger}\mathbf{g}_{m}\mathbf{s}+\sum_{m=1}^M\mathbf{g}_{m}^{\dagger}\mathbf{n}_m\\
=&\sqrt{\rho q_k}\sum_{m=1}^M(\mathbf{\hat{g}}_{m}^\mathrm{H}\mathbf{\hat{g}}_{m})^{-1}\mathbf{\hat{g}}_{m}^\mathrm{H}\left(\sqrt{1-\tau_D^2}\mathbf{\hat{g}}_m+\tau_D\mathbf{\Omega}_D\right)\mathbf{s}\\
&+\sum_{m=1}^M(\mathbf{\hat{g}}_{m}^\mathrm{H}\mathbf{\hat{g}}_{m})^{-1}\mathbf{\hat{g}}_{m}^\mathrm{H}\mathbf{n}_m\\
=&\sqrt{\rho q_k(1-\tau_D^2)}\sum_{m=1}^M\mathbf{s}+\tau_D\sqrt{\rho q_k}\sum_{m=1}^M(\mathbf{\hat{g}}_{m}^\mathrm{H}\mathbf{\hat{g}}_{m})^{-1}\mathbf{\hat{g}}_{m}^\mathrm{H}\mathbf{\Omega}_D\mathbf{s}\\
&+\sum_{m=1}^M(\mathbf{\hat{g}}_{m}^\mathrm{H}\mathbf{\hat{g}}_{m})^{-1}\mathbf{\hat{g}}_{m}^\mathrm{H}\mathbf{n}_m
\end{aligned}
\end{equation}
\footnote{Similar to \cite{2017-Ngo-CFMMIMO,2017-Nayebi-CFMMIMO}, the CANs are transmitted with the CPU through perfect communication links. These perfect communication links might be established through optical fiber between the CANs and CPU. Additionally, the communication links between CANs and CPU can be realized by copperbased backhaul links, which is capable of providing a capacity of $750$ Mbits/s with the maximum range of $1.5$ km according to \cite{2017-Imran-5G}.}

Denote $r_k$ as the aggregated received signal for TN $k$ in \eqref{CPUrece}, which can be expressed as
\begin{equation}\label{CPUk}
\begin{aligned}
r_k=&\sqrt{\rho(1-\tau_D^2)}\sum_{m=1}^M\sqrt{q_{k}}s_{k}+\\
&\tau_D\sqrt{\rho q_k}\sum_{m=1}^M(\mathbf{\hat{g}}_{m,k}^\mathrm{H}\mathbf{\hat{g}}_{m,k})^{-1}\mathbf{\hat{g}}_{m,k}^\mathrm{H}\mathbf{\Omega}_{D,k}s_k+\\
&\sum_{m=1}^M(\mathbf{\hat{g}}_{m,k}^\mathrm{H}\mathbf{\hat{g}}_{m,k})^{-1}\mathbf{\hat{g}}_{m,k}^\mathrm{H}n_{m,k}
\end{aligned}
\end{equation}
Consider the worst-case of the uncorrelated Gaussian noise, we can obtain the received SINR of the transmitted signal in \eqref{CPUk} \cite{2017-Ngo-CFMMIMO}, which is shown as \eqref{uplinkSINR} in the bottom of this page.
\begin{figure*}[hb]
\normalsize
\begin{equation}\label{uplinkSINR}
\gamma_k=\frac{\rho(1-\tau_D^2)|\sum_{m=1}^M\sqrt{q_{k}}|^2}{\rho\tau_D^2q_k|\sum_{m=1}^M(\mathbf{\hat{g}}_{m,k}^\mathrm{H}\mathbf{\hat{g}}_{m,k})^{-1}\mathbf{\hat{g}}_{m,k}^\mathrm{H}\mathbf{\Omega}_{D,k}s_k|^2+|\sum_{m=1}^M(\mathbf{\hat{g}}_{m,k}^\mathrm{H}\mathbf{\hat{g}}_{m,k})^{-1}\mathbf{\hat{g}}_{m,k}^\mathrm{H}n_{m,k}|^2}
\end{equation}
\hrulefill
\vspace*{10pt}
\end{figure*}
As cell-free massive MIMO systems support a very large number of CANs serving a small number of FANs, we characterize the asymptotic property of the received SINR in \eqref{uplinkSINR} with the cell-free massive MIMO setting in the following theorem,  which is shown in the bottom of the next page. By this theorem, we can solve the task allocation subproblem.
\newtheorem{theorem}{Theorem}
\begin{theorem}\label{tn1}
As the number of CANs goes to infinity, i.e., $M\rightarrow\infty$, the received SINR in \eqref{uplinkSINR} is capable of being asymptotically expressed as
\begin{equation}\label{infSINRo}
\gamma_k=\frac{q_k\rho(1-\tau_D^2)}{\frac{\sigma^2}{M}\sum_{m=1}^M\sum_{k=1}^{K}\lambda^{-1}_{k,m}}.
\end{equation}
\end{theorem}
\begin{proof}[Proof]
The proof is given in Appendix $A$.
\end{proof}

Based on Theorem \ref{tn1}, the received SINR for TN $k$ only related to transmit power from the $k$th FAN for very large $M$.

\subsection{Bandwidth Allocation}
As second-order derivative of the objective function is strictly negative, we have that $T_{\rm{total}}$ is a convex function of $\bm{\eta}$, so we obtain that optimization problem \eqref{originalP} is a convex optimization problem. Then, we refer to the Lagrangian multiplier method to solve problem \eqref{originalP}. We denote $\kappa$, $\nu$ and $\xi$ as the Lagrange multipliers associated with each constraint in problem \eqref{originalP}. Consequently, we have the Lagrangian function as
\begin{multline}
L(\bm{\eta},\bm{\kappa},\bm{\nu},\bm{\xi})=\Omega_{\rm{total}}(\bm{\eta})+\sum_{k=1}^K\kappa_k\eta_{k}+\sum_{k=1}^K\nu_k(1-\eta_{k})\\
+\sum_{k=1}^K\xi_k(\sum_{k=1}^K\eta_{k}-B).
\end{multline}
Then, we have
\begin{equation}\label{lagrange}
\begin{aligned}
\frac{\partial L(\bm{\eta},\bm{\kappa},\bm{\nu},\bm{\xi})}{\partial \eta_{k}}=&\left(\frac{-\mu_kp_{b,m}l_i}{\eta^2_kB\log_2\left(1+\frac{p_{b,m}d_{b,m}|h_{b,m}|^2}{\sigma_m^2}\right)}-\right.\\
&\left.\frac{(1-\mu_k)l_k}{\eta^2_kB\log_2\left(1+\frac{p_{b,m}d_{b,m}|h_{b,m}|^2}{\sigma_m^2}\right)}\right)+\kappa_k\\
&-\nu_k+\xi_k.
\end{aligned}
\end{equation}
Next, the Karush-Kuhn-Tucker (KKT) conditions can be expressed as
\begin{equation}\label{KKT1}
\frac{\partial L(\bm{\eta}^{\ast},\bm{\kappa},\bm{\nu},\bm{\xi})}{\partial \eta^{\ast}_{k}}=0,\quad\forall k\in\mathcal{K},
\end{equation}
\begin{equation}\label{KKT2}
\kappa_k\eta_{k}^{\ast}=0,\quad\nu_k(1-\eta_{k}^{\ast})=0,\quad\forall k\in\mathcal{K},
\end{equation}
\begin{equation}
\xi_k(\sum_{k=1}^K\eta_{k}^{\ast}-B)=0,\quad\forall k\in\mathcal{K},
\end{equation}
\begin{equation}
\sum_{k=1}^K\eta_{k}^{\ast}=B,\quad 0\leq \eta_{k}^{\ast}\leq 1,\quad\forall k\in\mathcal{K}.
\end{equation}
Based on \eqref{lagrange} and \eqref{KKT1}, we arrive at
\begin{equation}
\begin{aligned}
\xi_k=&\frac{1}{(\eta^\ast_k)^2}\left(\frac{(1-\mu_k)l_k}{B\log_2\left(1+\frac{p_{b}d_{b,m}|h_{b,m}|^2}{\sigma_m^2}\right)}+\right.\\
&\left.\frac{\mu_kp_{b}l_k}{B\log_2\left(1+\frac{p_{b}d_{b,m}|h_{b,m}|^2}{\sigma_m^2}\right)}\right)-\kappa_k+\nu_k.
\end{aligned}
\end{equation}
As a result, the optimal solution can be summarized as follows:
\begin{itemize}
\item  Under the condition that\begin{equation}
\begin{aligned}
\xi_k\geq&\frac{1}{(\eta^\ast_k)^2}\left(\frac{(1-\mu_k)l_k}{B\log_2\left(1+\frac{p_{b}d_{b,m}|h_{b,m}|^2}{\sigma_m^2}\right)}+\right.\\
&\left.\frac{\mu_kp_{b}l_k}{B\log_2\left(1+\frac{p_{b}d_{b,m}|h_{b,m}|^2}{\sigma_m^2}\right)}\right).
\end{aligned}
\end{equation} we have $\eta^{\ast}_{k}=0$, $\nu_k=0$ and $\kappa_k\geq 0$ based on \eqref{KKT2}.
\item   Under the condition that \begin{equation}
\begin{aligned}
\xi_k\leq&\frac{1}{(\eta^\ast_k)^2}\left(\frac{(1-\mu_k)l_k}{B\log_2\left(1+\frac{p_{b}d_{b,m}|h_{b,m}|^2}{\sigma_m^2}\right)}+\right.\\
&\left.\frac{\mu_kp_{b}l_k}{B\log_2\left(1+\frac{p_{b}d_{b,m}|h_{b,m}|^2}{\sigma_m^2}\right)}\right).
\end{aligned}
\end{equation} we have $\eta^{\ast}_{k}=1$, $\kappa_k=0$ and $\nu_k\geq 0$ based on \eqref{KKT2}.
\item  If $0<\eta^{\ast}_{k}<1$, we have $\nu_k=\kappa_k=0$ based on \eqref{KKT2}. Thus, we have
\begin{equation}
\begin{aligned}
\xi_k=&\frac{1}{(\eta^\ast_k)^2}\left(\frac{(1-\mu_k)l_k}{B\log_2\left(1+\frac{p_{b}d_{b,m}|h_{b,m}|^2}{\sigma_m^2}\right)}+\right.\\
&\left.\frac{\mu_kp_{b}l_k}{B\log_2\left(1+\frac{p_{b}d_{b,m}|h_{b,m}|^2}{\sigma_m^2}\right)}\right).
\end{aligned}
\end{equation}
Then, the optimal solution can be derived as
$\eta^{\ast}_{k}=\sqrt{\frac{(1-\mu_k)l_k}{\xi_kB\log_2\left(1+\frac{p_{b}d_{b,m}|h_{b,m}|^2}{\sigma_m^2}\right)}+\frac{\mu_kp_{b,m}l_k}{\xi_kB\log_2\left(1+\frac{p_{b}d_{b,m}|h_{b,m}|^2}{\sigma_m^2}\right)}}$, $\forall k\in \mathcal{K}$.
\end{itemize}

Since $0<\eta_{k}<1$, $\forall k$, the optimal bandwidth allocation results $\bm{\eta}^\ast$ of Problem \eqref{originalP} can be obtained as \eqref{optimalban}, which is shown in the bottom of the next page.
\begin{figure*}[hb]
\normalsize
\begin{equation}\label{optimalban}
\eta^{\ast}_{k}=\sqrt{\frac{(1-\mu_k)l_k}{\xi_kB\log_2\left(1+\frac{p_{b}d_{b,m}|h_{b,m}|^2}{\sigma_m^2}\right)}+\frac{\mu_kp_{b,m}l_k}{\xi_kB\log_2\left(1+\frac{p_{b,m}d_{b,m}|h_{b,m}|^2}{\sigma_m^2}\right)}},\quad \forall k.
\end{equation}
\hrulefill
\vspace*{10pt}
\end{figure*}

\subsection{Optimizing Task Allocation}
Since problem \eqref{originalP} is a mixed integer programming problem with respect to $\bm{\alpha}$, to obtain its globally optimal solution is NP-hard. Under the circumstances, the partially dualized method of \cite{2004-Boyd-CO} can be adopted to solve such problem. Then, we denote $\bm{\psi}=[\psi_1,\cdots,\psi_K]$ as the
dual variable for the constraint \eqref{five}. Consequently, the Lagrangian function can be expressed as
\begin{multline}
F\left(\bm{\alpha},\bm{\psi}\right)=\Omega_{\rm{total}}(\bm{\alpha})
+\sum\limits_{k\in\mathcal{K}}\psi_k\left(\sum\limits_{k\in\mathcal{K}}\alpha_{k}-K\right)\\
=\sum\limits_{k\in\mathcal{K}}\left\{\frac{(1-\mu_k)l_k}{\eta_kB\log_2\left(1+\frac{p_{b}d_{b,m}|h_{b,m}|^2}{\sigma_m^2}\right)}+\frac{C_k(1-\mu_k)(1-\alpha_{k})l_{k}}{f_k^F}\right.\\
\left.+\frac{(1-\mu_k)\alpha_kl_k}{B\log_2\left(1+\gamma_k\right)}+\frac{C_k(1-\mu_k)\alpha_{k}l_{k}}{f_k^C}+\right.\\
\left.\frac{\mu_kp_{b}l_k}{\eta_iB\log_2\left(1+\frac{p_{b}d_{b,m}|h_{b,m}|^2}{\sigma_m^2}\right)}+C_0\mu_k(1-\alpha_{k})l_{k}P_{\mathrm{CN}}\right.\\
\left.+\frac{\mu_kq_k\alpha_kl_k}{B\log_2\left(1+\gamma_k\right)}\right\}+\psi\left(\sum\limits_{k\in\mathcal{K}}\alpha_{k}-K\right)\\
=\sum\limits_{k\in\mathcal{K}}\left(\frac{(1-\mu_k)l_k}{B\log_2\left(1+\gamma_k\right)}+\frac{(1-\mu_k)C_kl_{k}}{f_k^C}-\frac{(1-\mu_k)C_kl_{k}}{f_k^F}+\right.\\
\left.\frac{\mu_kq_kl_k}{B\log_2\left(1+\gamma_k\right)}-C_0\mu_kl_{k}P_{\mathrm{CN}}+\psi\right)\alpha_k+K\psi\\
+\sum\limits_{k\in\mathcal{K}}\left(\frac{(1-\mu_k)l_k}{\eta_kB\log_2\left(1+\frac{p_{b}d_{b,m}|h_{b,m}|^2}{\sigma_m^2}\right)}+\frac{C_kl_{k}}{f_k^F}\right.\\
\left.+\frac{\mu_kp_{b}l_k}{\eta_iB\log_2\left(1+\frac{p_{b}d_{b,m}|h_{b,m}|^2}{\sigma_m^2}\right)}\right).
\end{multline}

In this case, we can obtain the the partially dualized problem of the original Problem~\eqref{originalP}, which is given by
\begin{equation}\label{dual}
\begin{array}{ll}
f(\bm{\alpha},\bm{\psi})=\min\limits_{\alpha_{j}} \quad  F\left(\{\alpha_{j}\},\bm{\psi}\right),\\
\text{s.t.} \begin{array}[t]{rcl}
     \alpha_{j}&\in &  \{0,1\}, \forall j\in\mathcal{K},\\
      \sum_{j=1}^K\alpha_{j}&\leq & K.
       \end{array}
\end{array}
\end{equation}
Then, the analytical task allocation results can be explicitly obtained for fixed dual variable $\psi_k$, which is shown as \eqref{offs1} in the bottom of the next page. As $k^\ast$ is not unique, the tasks will be computed remotely at CAN with $k^\ast=1$.
\begin{figure*}[hb]
\normalsize
\begin{equation}\label{offs1}
\alpha^{\ast}_{k}=\left\{\begin{matrix}
1, & \text{if}~k^\ast=\arg\min\limits_{k'\in\mathcal{K}}(\frac{(1-\mu_{k'})l_{k'}}{B\log_2\left(1+\gamma_{k'}\right)}+\frac{(1-\mu_{k'})C_{k'}L_{k',s}}{f_{k'}^C}-\frac{(1-\mu_{k'})C_{k'}l_{k'}}{f_{k'}^F}+\\
&\frac{\mu_{k'}q_{k'}l_{k'}}{B\log_2\left(1+\gamma_{k'}\right)}-C_0\mu_{k'}l_{k'}P_{\mathrm{CN}}+\psi),\\
0, & \text{otherwise},
\end{matrix}\right.
\end{equation}
\hrulefill
\vspace*{10pt}
\end{figure*}
The optimal dual variable $\psi_k^{\ast}$ can be obtained by applying the subgradient approach of~\cite{2004-Boyd-CO} or the coordinate descent approach \cite{2014-Yu-price}. According to the above results, we have a theorem shown in the following.
\begin{theorem}\label{tn2}
For binary task offloading in cell-free massive MIMO systems, when the number of CANs tends to infinity, i.e., $M\rightarrow\infty$, all the tasks will be computed remotely at CAN for the delay tolerant services, i.e., $\mu_{k}=1$, $\forall k$. On the other hand, only the tasks from node $i^\ast$ will be computed remotely at CAN, where FAN $i^\ast$ has the lest computational capacity. The other tasks will be computed locally at the FAN for the delay sensitive services, i.e., $\mu_{k}=0$, $\forall k$.
\end{theorem}
\begin{proof}[Proof]
The proof is given in Appendix $B$.
\end{proof}

\section{Performance Evaluation}
In this section, our simulation results characterize the performance of bandwidth allocation and task allocation for the proposed delay-aware cell-free massive MIMO-aided multi-tier computing systems compared to several baseline schemes.

\subsection{System Parameters}
\textcolor{black}{To verify the accuracy of the analysis results, we run simulations over $10,000$ channel realizations to obtain the averaging results. Additionally, the channel coefficient samples are generated at a period of $0.005$ms during the simulation.} Unless otherwise noted, most simulations follow the following scenario. The cell-free massive MIMO systems have $25$ FANs with sufficient computing resources. 
The computational capacity of each CN is selected from the set $\{0.2,0.3,\cdots,0.8\}$ GHz randomly and will be fixed. 
The local computing energy per computation cycle $z_{i}$ follows a uniform distribution in the range of $(0,20\times 10^{-11})$ J/cycle. For the computational task, we consider the service similar to that in \cite{2016-Chen-Offloading}, where any task $k_s$ has a size of $l_k=500$ KB, $\forall k\in\mathcal{S}$, and the required computation cycles per bit follows a uniform distribution in the range of $[500,1500]$ cycles/bit.

\begin{figure}[!t]
\begin{center}
\includegraphics [width=3.6in]{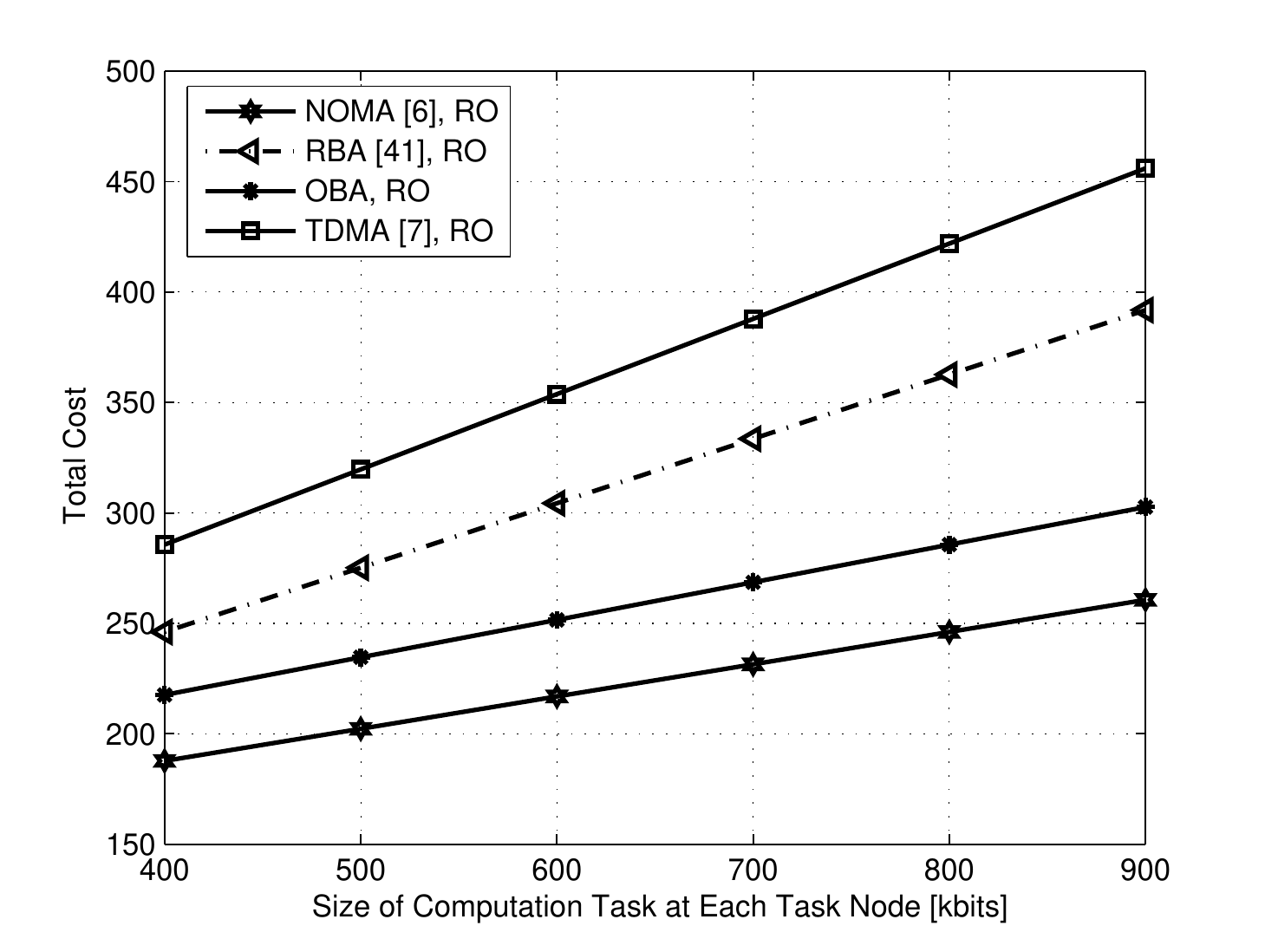}
\caption{\textcolor{black}{Total task computation cost versus the task size, for $K=8$.}}\label{fig2}
\end{center}
\end{figure}
\begin{figure}[!t]
\begin{center}
\includegraphics [width=3.6in]{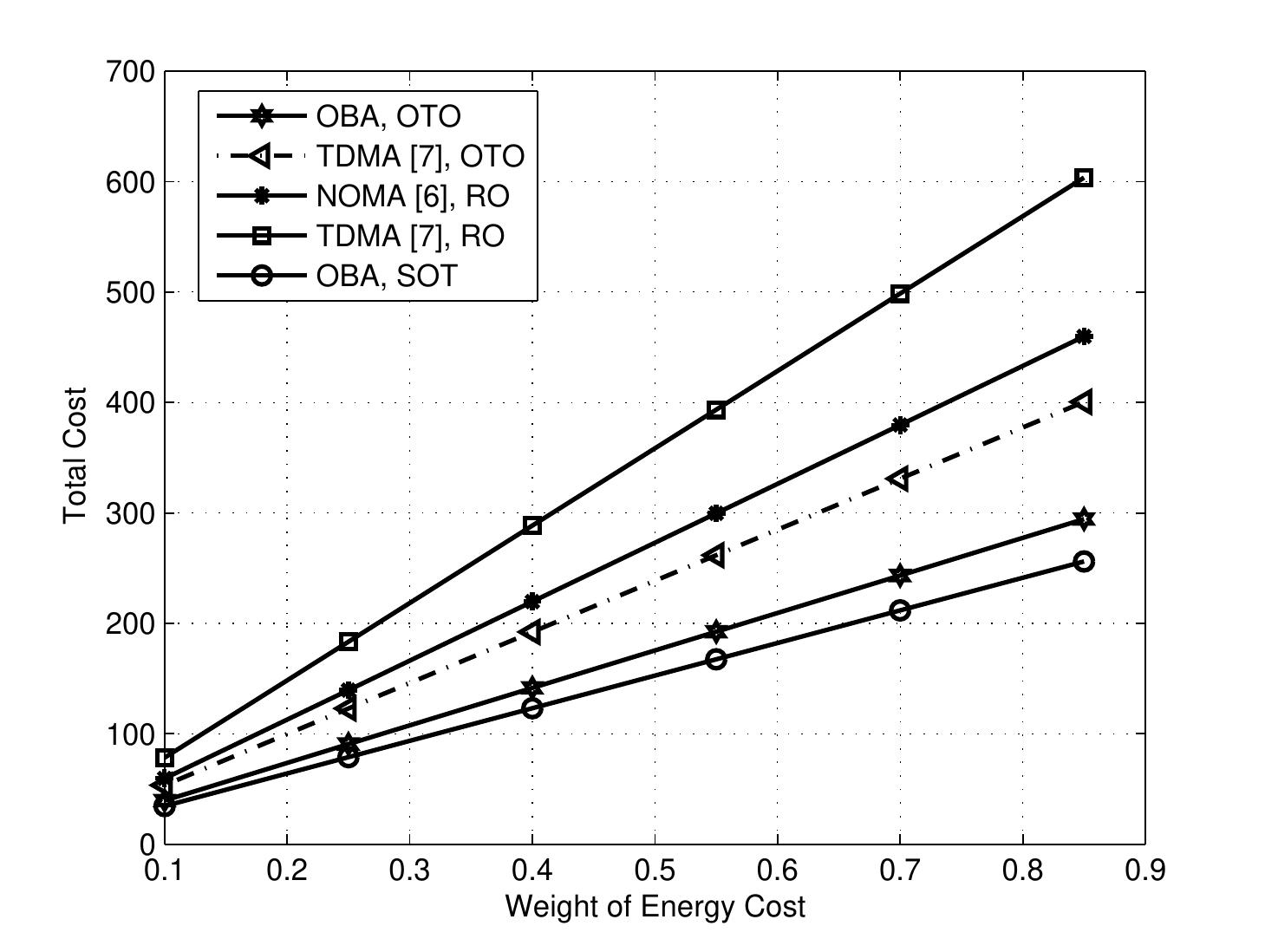}
\caption{\textcolor{black}{Total cost versus the weight of energy cost $\mu_k$, $\forall k$, for $K=8$.}}\label{fig3}
\end{center}
\end{figure}
\begin{figure}[!t]
\begin{center}
\includegraphics [width=3.6in]{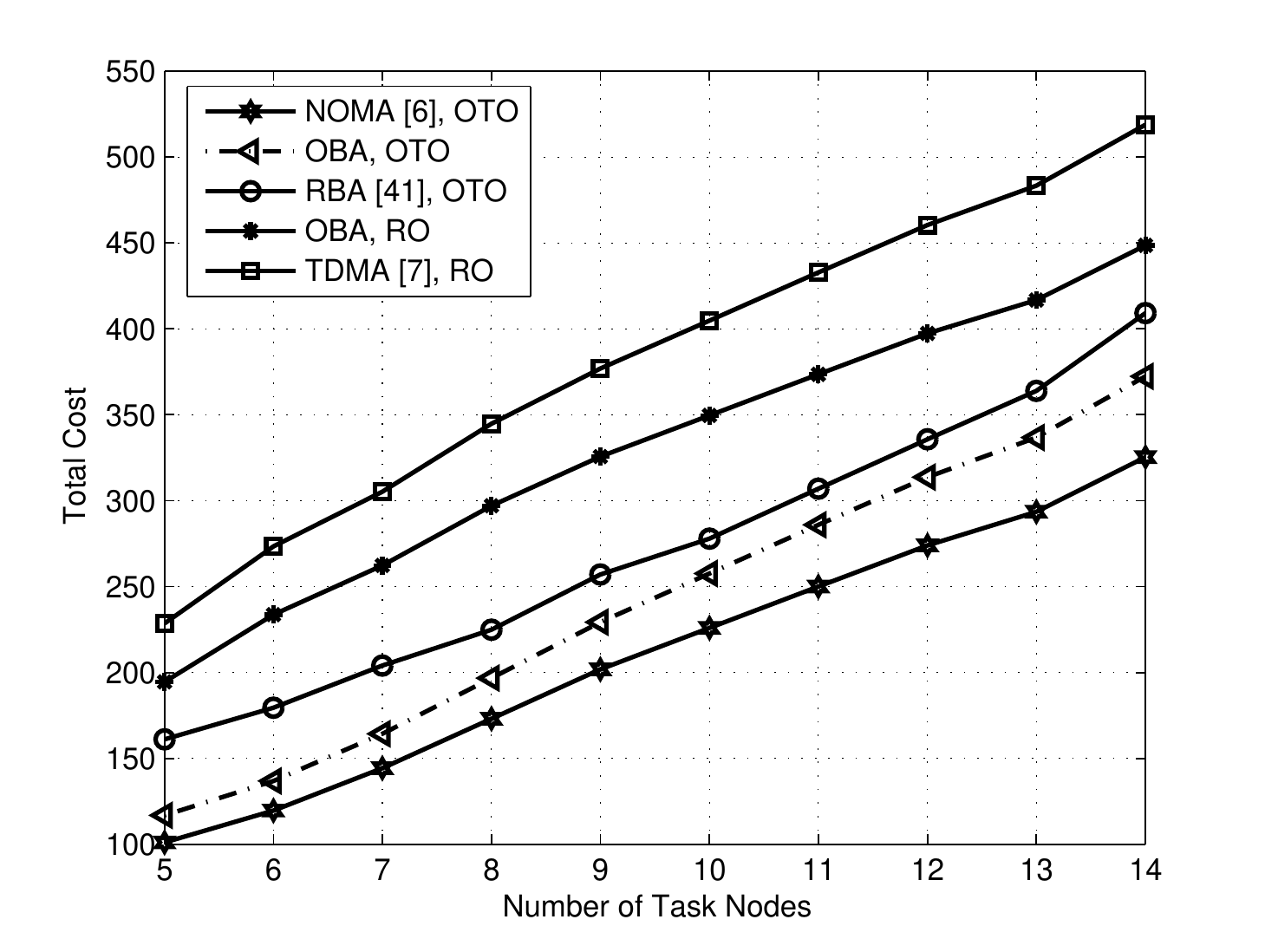}
\caption{\textcolor{black}{Total cost versus the number of TNs $K$.}}\label{fig4}
\end{center}
\end{figure}

\begin{figure}[!t]
\begin{center}
\includegraphics [width=3.6in]{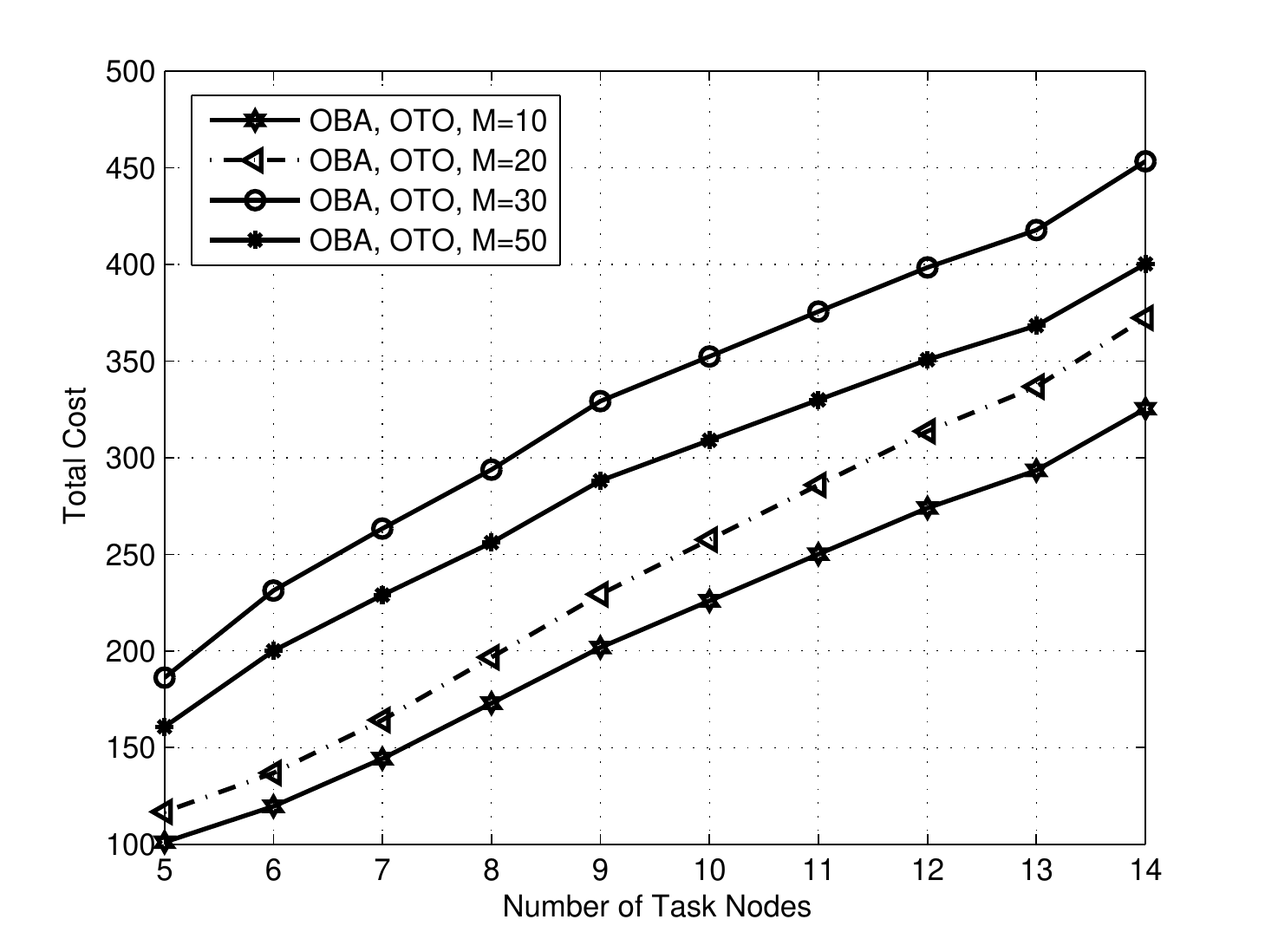}
\caption{\textcolor{black}{Total cost versus different number of TNs $K$ for different number of CANs.}}\label{fig5}
\end{center}
\end{figure}

\begin{figure}[!t]
\begin{center}
\includegraphics [width=3.6in]{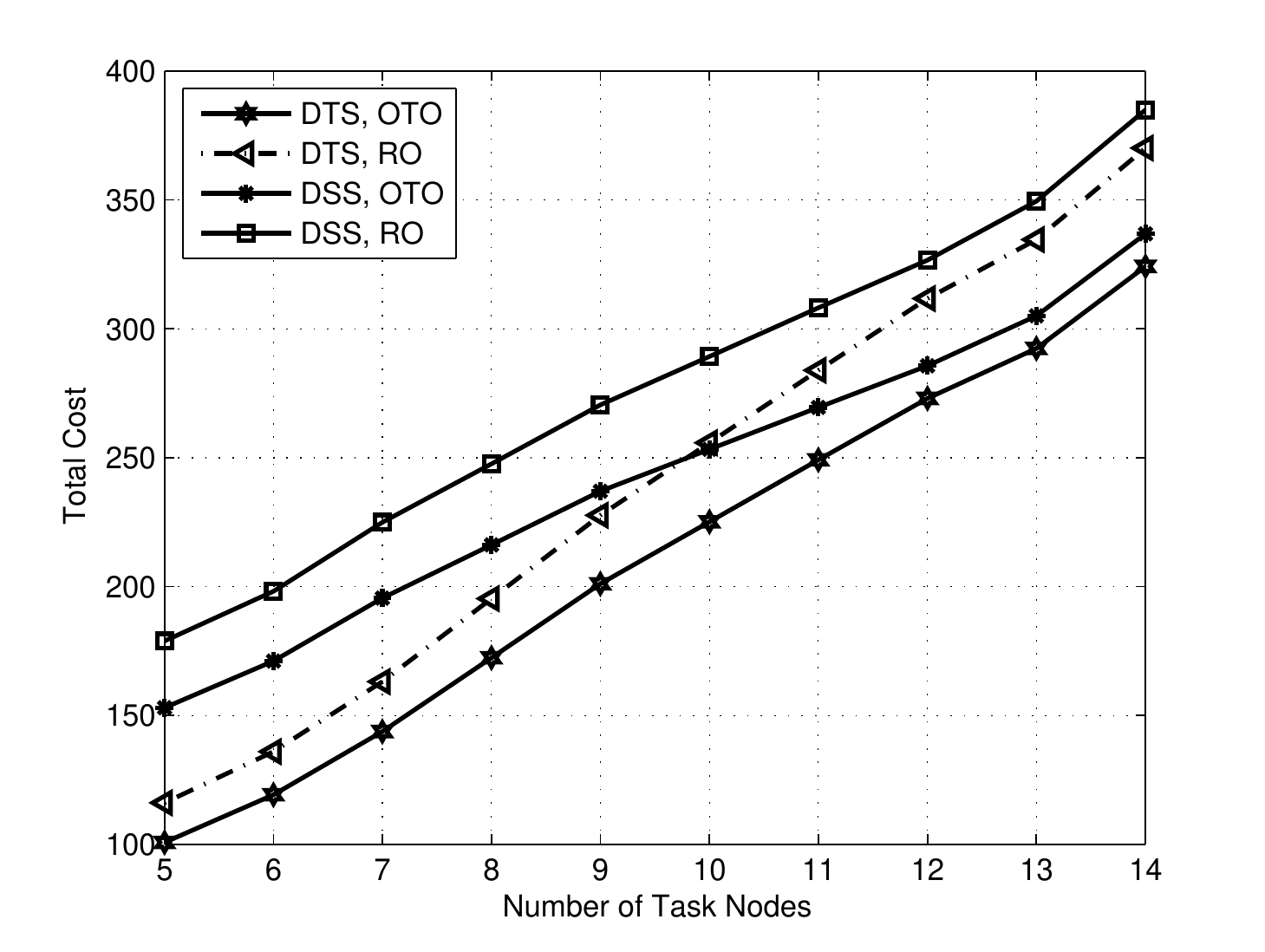}
\caption{\textcolor{black}{Total cost versus different number of TNs for different delay requirement services.}}\label{fig6}
\end{center}
\end{figure}

\subsection{Performance bandwidth allocation}
\textcolor{black}{In Fig.~\ref{fig2}, we show the total cost under different sizes of the computational task at each TN,} \textcolor{black}{where we compare our optimal bandwidth allocation (OBA) strategy to the conventional techniques of time division multiple access (TDMA) \cite{2018-Yang-MEETS}, non-orthgonal multiple access (NOMA) \cite{2020-Wang-NOMAFOG} and random bandwidth allocation (RBA) \cite{2021-Song-NOMAedge}. Explicitly, in the NOMA strategy, each TN offloads the task with the same time and frequency.} It can be observed from this figure that the NOMA performs best in total energy and delay cost compared to TDMA and FDMA. This is because the task offloading latency can be largely reduced by NOMA strategy. However, as the received signals from different TNs are superimposed at the FANs, this strategy has relatively highest computational complexity than TDMA and FDMA.  \textcolor{black}{As expected, we can observe from this figure that our OBA strategy always performs best in total cost than those of the TDMA and RBA strategies for different sizes of computational task.} Furthermore, we can observe that the NOMA offers better performance in total cost than that of the OBA strategy, since OBA strategy is derived from FDMA. Upon increasing the size of the computational task, there is an another interesting remark, the total cost is increased. If the number of size of the computational task is higher, then the task computation energy consumption will be higher. Additionally, it can be observed from this figure that the OBA strategy offers better performance than TDMA in terms of total energy and delay consumption, implying that it is more beneficial for optimal bandwidth allocation in total energy and delay consumption. 

For the comparison of different bandwidth allocation strategies, we also plot the total cost when applying simulated optimal task offloading (SOT), optimized task offloading (OTO) and random task offloading (RO) strategies. In Fig.~\ref{fig3}, we show the total task offloading cost versus the weight of energy consumption $\mu_k$, $\forall k$, for different task offloading strategies. \textcolor{black}{We can observe from this figure that the total cost increases when increasing $\mu_k$, implying that the total energy consumption dominates the total cost. It can be observed that the total task offloading cost of OBA strategy is lower than that of the TDMA under the condition of the same task offloading policy. These observations further verify that the system cost can be reduced by different weights of energy consumption with OBA policy.} In summary, OBA strategy makes the task offloading more power efficient without degrading corresponding task offloading rate. \textcolor{black}{As expected, it can be observed from this figure that SOT always performs best in total cost than those of the OTO and RO strategies for different weights of energy consumption.}
Furthermore, we can observe that the OTO strategy always reduces the total cost regardless of the bandwidth allocation strategies, which confirms the analytical results. In case of energy efficient task offloading, our simulation results further indicate that taking bandwidth allocation into account is necessary, which can realize energy-efficient task offloading for different delay requirements.

\subsection{Performance of task allocation}
\textcolor{black}{As shown analytically in Section \uppercase\expandafter{\romannumeral4}, the cell-free massive MIMO gain of the bandwidth allocation strategy over the task offloading strategy decreases with the number of CAN as well as the received SINR, while it is not affected by these parameters when the number of CAN goes to infinity. In this section, we demonstrate the impact of $M$, $K$, and different requirements of the computational tasks on the total cost by means of simulations.}

\textcolor{black}{Fig.~\ref{fig4} illustrates the total cost versus different number of TNs for different bandwidth allocation and task offloading policies. We compare our strategy to traditional systems operating with RO, such that RO serves as a benchmark. As expected, we can observed from this figure that our OTO strategy always performs better in total energy and delay cost than that of the RO policy with different number of TNs. Additionally, it can observed that the OTO scheme relying on the optimal bandwidth allocation strategy can always achieve better performance than that of the traditional offloading scheme. Furthermore, as expected, optimal bandwidth allocation scheme offers much better performance than TDMA and RBA utilizing the same task allocation scheme. It should be noted that even if the OBA scheme is adopted, OTO performs better by the optimal task allocation scheme with the aid of cell-free massive MIMO. This is due to the fact that the optimal task allocation scheme increases the task offloading throughput as well as decreases its energy cost of the cell-free massive MIMO systems. On the other hand, it can be observed that the total cost is increased by increasing the number of TNs. It should be noted that if the number of TNs is larger, then the energy cost and total computational delay of the task offloading cell-free massive MIMO systems will be larger. }

Fig.~\ref{fig5} illustrates the performance of the total cost in terms of both different number of TNs and different number of CANs. As expected, it can be observed that the OBA scheme using OTO strategy can reduce total cost. Meanwhile, the total cost is increased by increasing the number of TNs. This result implies that the total energy consumption and computational latency are increased. Furthermore, we can observe that the larger the number of CAN is, the smaller the total cost is. This is because the larger the number of CAN, the higher the SINR at the CPU, thus reducing both the total task offloading delay and energy cost. As the optimal task allocation strategy can realize energy-efficient task transmission, the larger the number of CAN utlizing the optimal task allocation strategy can always perform better in total cost, which confirms our analysis results. It should be noted that these results indicate that the total cost is decreased by placing more CANs.

Next, Fig.~\ref{fig6} plots the total cost versus different number of TNs $K$ for different delay requirement services. It can be observed that the total cost is significantly increased upon increasing $K$ and the delay tolerant services (DTS) with OTO always performs best in total cost, which is due to the fact that the tasks are allocated to the CPU with cloud computation for DTS with OTO. Additionally, the OTO scheme associated with optimal task allocation scheme always offers better performance than the RO scheme. This is because the optimal task allocation scheme can achieve the minimal delay and energy consumption for task transmission. We can also observe that for different number of TNs, delay sensitive services (DSS) always have much larger computational latency and energy consumption than DTS for the same task offloading strategy. Since the DSS has much more stringent delay requirement, result in larger energy consumption and total cost, which confirms our analysis results. We have another interesting remark from this figure, there is a crossing point by increasing the number of TNs between the performances of the DSS and the DTS for different task offloading strategies. According to these simulation results, when $K$ is small, the DTS with RO strategy has smaller total cost than DSS with OTO scheme. On the other hand, when $K$ is large, the DTS with RO strategy has larger total cost than DSS with OTO scheme. We hold the opinion that the optimal task offloading strategy has a great influence on the total cost when $K$ is large. In this case, the OTO will always performs better regardless of the different delay requirements of the tasks.

\section{Conclusions}
In this paper, we proposed a cell-free massive MIMO-assisted multi-tier computing systems and investigated the bandwidth allocation and task offloading, where the intensive tasks from TNs can be offloaded to nearby FAN, and to the CPU constituted by the nearby CAN via the cell-free massive MIMO. \textcolor{black}{We formulated a total cost minimization problem in terms of energy consumption and computational latency}, while considering realistic heterogenous delay requirements of the tasks. Since we consider binary task offloading, the resulting non-convex bandwidth and task allocations problem can be solved \textcolor{black}{by decoupling the original problem}. As the bandwidth allocation problem is a convex optimization problem, we first obtained the bandwidth allocation solution under a given task allocation strategy, followed \textcolor{black}{by conceiving the traditional convex optimization method to determine the bandwidth allocation result}. According to the bandwidth allocation solution, the Lagrange partial relaxation method has been used for formulating the Lagrange dual problem by relaxing the binary constraint to determine the task offloading result. \textcolor{black}{The simulation results demonstrate that the proposed
strategy always performs best with the benchmark strategies. Meanwhile, based on the received SINR obtained by $M\rightarrow\infty$, the cost optimal task offloading strategy can be chosen for heterogeneous delay requirements of the computational tasks in our proposed cell-free massive MIMO-assisted multi-tier computing systems.} \textcolor{black}{Future work is in progress to consider more general case that arbitrary FANs are deployed for each TN to assist task offloading in the proposed framework.}

\appendices
\section{proof of Theorem~\ref{tn1}}\label{a1}
According to Tchebyshev¡¯s theorem \cite{1970-Cramer-RV}, we have
\begin{equation}\small
\lim_{N\rightarrow \infty}\frac{1}{N}(X_1+X_2+\ldots+X_N)=\frac{1}{N}(\textsf{E}(X_1)+\textsf{E}(X_2)+\ldots+\textsf{E}(X_N)),
\end{equation}
where $X_1$, $X_2$, $\ldots$, $X_N$ are N independent random variables, $\textsf{E}(X_i)$ denotes the expectation of $X_i$, $\forall i$.
Then, regarding the second term on the right hand side in \eqref{CPUk}, we have
\begin{equation}\label{interference}
\lim_{M\rightarrow \infty}\frac{1}{M}\tau_D\sqrt{\rho q_k}\sum_{m=1}^M(\mathbf{\hat{g}}_{m,k}^\mathrm{H}\mathbf{\hat{g}}_{m,k})^{-1}\mathbf{\hat{g}}_{m,k}^\mathrm{H}\mathbf{\Omega}_{D,k}s_k=0.
\end{equation}

By adopting the eigenvalue/eigenvector decomposition of $\mathbf{\hat{g}}_{m,k}^\mathrm{H}\mathbf{\hat{g}}_{m,k}$, we obtain
\begin{equation}
\mathbf{\hat{g}}_{m,k}^\mathrm{H}\mathbf{\hat{g}}_{m,k}=\mathbf{Q}\mathbf{\Lambda }\mathbf{Q}^\mathrm{H},
\end{equation}
where $\mathbf{\Lambda}=\rm{diag}\{\lambda_1,\cdots,\lambda_K\}$ and $\mathbf{Q}$ respectively denotes the nonnegative diagonal eigenvalue matrix and the unitary eigenvector matrix, respectively. Thus, regarding the third term on the right hand side in \eqref{CPUk}, we have
\begin{equation}
\begin{aligned}
&\lim_{M\rightarrow \infty}\frac{1}{M}\sum_{m=1}^M(\mathbf{\hat{g}}_{m,k}^\mathrm{H}\mathbf{\hat{g}}_{m,k})^{-1}\mathbf{\hat{g}}_{m,k}^\mathrm{H}n_{m,k}\\
=&\lim_{M\rightarrow \infty}\frac{1}{M}\sum_{m=1}^M(\mathbf{Q}\mathbf{\Lambda }\mathbf{Q}^\mathrm{H})^{-1}\mathbf{\hat{g}}_{m,k}^\mathrm{H}n_{m,k}\\
=&\lim_{M\rightarrow \infty}\frac{1}{M}\sum_{m=1}^M(\mathbf{\Lambda })^{-1}\mathbf{\hat{g}}_{m,k}^\mathrm{H}n_{m,k}
\end{aligned}
\end{equation}

Then, we have
\begin{equation}\label{noise}
\begin{aligned}
&\lim_{M\rightarrow \infty}\frac{1}{M}|\sum_{m=1}^M(\mathbf{\hat{g}}_{m,k}^\mathrm{H}\mathbf{\hat{g}}_{m,k})^{-1}\mathbf{\hat{g}}_{m,k}^\mathrm{H}n_{m,k}|^2\\
=&\lim_{M\rightarrow \infty}\frac{\sigma^2}{M}\sum_{m=1}^M\rm{tr}\left((\mathbf{\hat{g}}_{m,k}^\mathrm{H}\mathbf{\hat{g}}_{m,k})^{-1}\mathbf{\hat{g}}_{m,k}^\mathrm{H}\mathbf{\hat{g}}_{m,k}(\mathbf{\hat{g}}_{m,k}^\mathrm{H}\mathbf{\hat{g}}_{m,k})^{-1}\right)_{k,k}\\
=&\frac{\sigma^2}{M}\sum_{m=1}^M\sum_{k=1}^{K}\lambda^{-1}_{k,m}.
\end{aligned}
\end{equation}
Based on \eqref{interference} and \eqref{noise}, we have
\begin{equation}
\lim_{M\rightarrow \infty}\gamma_k=\frac{q_k\rho(1-\tau_D^2)}{\frac{\sigma^2}{M}\sum_{m=1}^M\sum_{k=1}^{K}\lambda^{-1}_{k,m}}.
\end{equation}

\section{proof of Theorem~\ref{tn2}}\label{a2}
While all tasks are from delay tolerant services, the energy consumption will dominate the total cost. Then, $\mu_{k}=1$, $\forall k$. Substitute $\mu_{k}=1$ into \eqref{offs1}, we have
\begin{equation}
\alpha^{\ast}_{k}=\left\{\begin{matrix}
1, & \text{if}~k^\ast=\arg\min\limits_{k'\in\mathcal{K}}(\frac{ql}{B\log_2\left(1+\gamma_{k'}\right)}-C_0lP_{\mathrm{CN}}+\psi),\\
0, & \text{otherwise}.
\end{matrix}\right.
\end{equation}
As a result, $\alpha^{\ast}_{k}=1$, $\forall k$. Therefore, all the tasks will be computed remotely at CAN.

On the other hand, while all the tasks are from delay sensitive services, the total delay cost will dominate the total cost. Then, $\mu_{k}=0$, $\forall k$. Substitute $\mu_{k}=1$ into \eqref{offs1}, we have
\begin{equation}
\alpha^{\ast}_{k}=\left\{\begin{matrix}
1, & \text{if}~k^\ast=\arg\min\limits_{k'\in\mathcal{K}}(\frac{l}{B\log_2\left(1+\gamma_{k'}\right)}+\frac{C_{k'}l}{f_{k'}^C}-\frac{C_{k'}l}{f_{k'}^F}+\psi),\\
0, & \text{otherwise}.
\end{matrix}\right.
\end{equation}
Then, only the tasks from node $i^\ast$ will be computed remotely at CAN, where $i^\ast=\arg\min\limits_{i\in\mathcal{K}}(\frac{l_{i}}{B\log_2\left(1+\gamma_{i}\right)}-\frac{C_{i}l_{i}}{f_{i}^F})$, i.e., $i^\ast$ is from the smallest $f_{k}^F$, $\forall k$. This indicates that the FAN with the lest computational capacity will offload the tasks to the CAN, while the other tasks will be computed locally at the FAN.


\bibliography{mybib}
\bibliographystyle{ieeetr}

\end{document}